\documentclass[leqno,a4paper,twoside,11pt]{article}
\usepackage[nointlimits,nosumlimits]{amsmath}
\usepackage{amsfonts,amssymb,amsthm,ifthen}
\usepackage{bbm}
\usepackage[arrow,matrix,curve]{xy}
\usepackage[utf8x]{inputenc}
\usepackage[square,numbers]{natbib}

\newtheorem{Thm}{Theorem}[section]
\newtheorem{Prp}[Thm]{Proposition}
\newtheorem{Lem}[Thm]{Lemma}
\newtheorem{Cor}[Thm]{Corollary}
\newtheorem{Def}[Thm]{Definition}
\newtheorem{Rem}[Thm]{Remark}

\theoremstyle{definition}

\newcommand{\scal}[3][]{\ifthenelse{\equal{#1}{}}{
  \left\langle #2,\,#3 \right\rangle
}{\ifthenelse{\equal{#1}{(}}{
  \left( #2,\,#3 \right)
}{\ifthenelse{\equal{#1}{[}}{
  \left[ #2,\,#3 \right]
}{
  #1\left( #2,\,#3 \right)
}}}}

\newcommand{\arr}[2]{%
  \begin{array}{@{}#1@{}}#2\end{array}}
\newcommand{\abs}[1]{\left| #1 \right|}
\newcommand{\dd}[2][]{\frac{\partial #1}{\partial #2}}

\newcommand{\mA}{\mathcal A}
\newcommand{\mD}{\mathcal D}
\newcommand{\mO}{\mathcal O}
\newcommand{\mS}{\mathcal S}
\newcommand{\bN}{\mathbb N}
\newcommand{\bR}{\mathbb R}

\renewcommand{\div}{\mathrm{div}}
\newcommand{\Der}{\mathrm{Der}}
\newcommand{\Hom}{\mathrm{Hom}}
\newcommand{\dsvol}{\mathrm{dsvol}}
\newcommand{\dvol}{\mathrm{dvol}}
\newcommand{\ev}{\mathrm{ev}}
\newcommand{\id}{\mathrm{id}}
\newcommand{\sdet}{\mathrm{sdet}}
\newcommand{\spann}{\mathrm{span}}
\newcommand{\str}{\mathrm{str}}
\newcommand{\tr}{\mathrm{tr}}

\newcommand{\laplace}{\triangle}
\newcommand{\vx}{\vec x}

\renewcommand{\title}[1]{\vbox{\center\LARGE{\textsc{#1}}}\vspace{5mm}}
\renewcommand{\author}[1]{\vbox{\center\large{\textsc{#1}}}\vspace{5mm}}
\newcommand{\address}[1]{\vbox{\center\em#1}}
\newcommand{\email}[1]{\vbox{\center\tt#1}\vspace{5mm}}

\begin{document}

\title{Divergence Theorems and the Supersphere}

\author{Josua Groeger$^1$}

\address{Humboldt-Universit\"at zu Berlin, Institut f\"ur Mathematik,\\
  Rudower Chaussee 25, 12489 Berlin, Germany }

\email{$^1$groegerj@mathematik.hu-berlin.de}

\begin{abstract}
\noindent
The transformation formula of the Berezin integral holds, in the non-compact case, only up to boundary integrals,
which have recently been quantified by Alldridge-Hilgert-Palzer.
We establish divergence theorems in semi-Riemannian supergeometry by means of the flow of vector fields
and these boundary integrals, and show how superharmonic functions are related to conserved quantities.
An integration over the supersphere was introduced by Coulembier-De Bie-Sommen as a generalisation of the Pizzetti integral.
In this context, a mean value theorem for harmonic superfunctions was established.
We formulate this integration along the lines of the general theory and
give a superior proof of two mean value theorems based on our divergence theorem.
\end{abstract}

\section{Introduction}

The analogon of the classical integral transformation formula in supergeometry holds
for compactly supported quantities only.
The nature of the additional boundary terms occurring in the non-compact case was only recently studied in \cite{AHP12},
where it was observed that a global Berezin integral can be defined by introducing a retraction as an
additional datum.

In the present article, we establish divergence theorems in semi-Riemannian supergeometry
by means of the flow of vector fields as studied in \cite{MSV93} together with the main result of \cite{AHP12}
concerning the change of retractions. We focus on the non-degenerate case, but also yield
a divergence theorem for a degenerate boundary, thus generalising a result in \cite{Uen95}.
While the lack of boundary compatibility conditions in general leads to the aforementioned boundary terms,
we show, moreover, that divergence-free vector fields are conserved quantities in a very natural sense
for any boundary supermanifold.
As shown in \cite{Gro13}, such vector fields arise from Killing vector fields via superharmonic maps.
We apply that theory to superharmonic functions.

The supersphere occurs naturally in the theory of the supersymmetric quantum hall effect (\cite{Has08}, \cite{HT13})
and underlies certain field theories, see \cite{SW05}.
In a series of papers (\cite{DBS07}, \cite{CDS09}, \cite{Cou12}), an integration over the supersphere was introduced,
first by an extension of Pizzetti's formula to super-polynomials, then by a formula for general superfunctions
later expressed in terms of an embedding, while a particular case of this integral was already studied in \cite{Jar88}.
In this context, \cite{CDS10} established a super analogon of the classical mean value theorem for harmonic functions.

The second purpose of the present article is to formulate this supersphere integration
in terms of a retraction $\gamma$ along the lines of the general theory and finally to give
a new self-contained proof of two mean value theorems for harmonic functions based on our first divergence theorem,
thus avoiding the subtle points left open in the proof of \cite{CDS10}.
As a corollary, we yield a simple expression for the boundary term concerning the change of retractions
from $\gamma$ to the standard retraction.

\section{Integration on Semi-Riemannian Supermanifolds}
\label{secIntegration}

In this section, we will briefly recall elements of the theory of semi-Riemannian supermanifolds
and integration of Berezinian forms, with a detour to the divergence of vector fields.

Throughout the article, we consider supermanifolds and their morphisms in
the sense of Berezin-Kostant-Leites as in \cite{Lei80}.
A supermanifold is thus, in particular, a ringed space $(M,\mO_M)$,
and a morphism $\varphi:(M_0,\mO_M)\rightarrow(N_0,\mO_N)$ of supermanifolds
consists of two parts $\varphi=(\varphi_0,\varphi^{\sharp})$ with $\varphi^{\sharp}$
denoting a generalised pullback of superfunctions $f\in\mO_N$.
Modern monographs on the general theory of supermanifolds include \cite{Var04} and \cite{CCF11}.
Following the conventions of \cite{Gro11a} and \cite{Gro13}, we denote the (super) tangent sheaf of $M$,
i.e. the sheaf of superderivations of $\mO_M$, by $\mS M:=\Der(\mO_M)$. A vector field $X\in\mS M$ can,
in local coordinates $x$, be expanded as
\begin{align}
\label{eqnVectorField}
X=\dd{x^k}\cdot X^k\qquad\textrm{such that}\qquad X(x^j)=(-1)^{\abs{x^j}\abs{X^j}}X^j
\end{align}
An even superfunction $f\in\mO_N$ is called positive if the underlying function $f_0\in C^{\infty}(M_0)$ is.
In this case, $f$ has a unique positive square root $\sqrt{f}$. Considering $X(\sqrt{f}\cdot\sqrt{f})$
for a vector field $X$, we obtain the calculation rule
\begin{align}
\label{eqnCalculationRule}
X(\sqrt{f})=\frac{1}{2}\sqrt{f}^{-1}X(f)\qquad\textrm{and analogous for fractional powers of $f$}.
\end{align}
The differential of a morphism $\varphi:M\rightarrow N$
is the morphism of sheaves $d\varphi:(\varphi_0)_*\mS M\rightarrow\mS\varphi$ defined by $d\varphi(X):=X\circ\varphi^{\sharp}$,
where $\mS\varphi:=\Der(\mO_N,(\varphi_0)_*\mO_M)$ denotes the sheaf of derivations (vector fields) along $\varphi$.
A bilinear form $B\in\Hom_{\mO_N}(\mS N\otimes\mS N,\mO_N)$ is pulled back under $\varphi$ to a bilinear form
on $M$ as follows.
\begin{align}
\label{eqnPullback}
\scal[\varphi^*B]{X}{Y}:=\scal[B_{\varphi}]{d\varphi[X]}{d\varphi[Y]}\;,\qquad
\scal[B_{\varphi}]{\varphi^{\sharp}\circ A}{\varphi^{\sharp}\circ B}:=\varphi^{\sharp}\circ\scal[B]{A}{B}
\end{align}
For consistent notation, we also prescribe $\varphi^*(f):=\varphi^{\sharp}(f)$ for a superfunction $f\in\mO_N$.

A semi-Riemannian supermetric $g$ on $M$ is an even, nondegenerate and supersymmetric bilinear form.
It follows that the odd part of the dimension of a semi-Riemannian supermanifold $(M,g)$ is an even number
which, as in \cite{Cou12}, we shall denote by $\dim M=m|2n$.
Moreover, $M$ possesses a unique connection $\nabla$ which is torsion-free and metric with respect to $g$,
called the Levi-Civita connection \cite{Goe08}.

\begin{Lem}
\label{lemChristoffel}
Let $x=(x^1,\ldots,x^{m+2n})$ be local coordinates of $M$ and denote by $\partial_i:=\partial_{x^i}$ the induced
local vector fields. Then the Christoffel symbols for the Levi-Civita connection defined by
$\Gamma_{ij}^l\partial_l=\nabla_{\partial_i}\partial_j$ are
\begin{align*}
\Gamma_{ij}^m=\frac{1}{2}\left(\partial_ig_{jk}+(-1)^{\abs{i}\abs{j}}\partial_jg_{ik}-(-1)^{\abs{k}(\abs{i}+\abs{j})}\partial_kg_{ij}\right)g^{km}
\end{align*}
where the matrix element $g^{km}$ is defined by $g_{lk}g^{km}=\delta_l^m$. It satisfies
\begin{align*}
\abs{g^{km}}=\abs{k}+\abs{m}\quad\mathrm{and}\qquad g^{km}=(-1)^{\abs{k}\abs{m}+\abs{k}+\abs{m}}g^{mk}
\end{align*}
\end{Lem}

\begin{proof}
The Christoffel symbols are obtained by a standard calculation e.g. as in Prp. 2.12 of \cite{GW12b}.
Moreover, the parity of $g^{km}$ as stated follows at once from the evenness of $g^{-1}$. Finally,
the last equality holds by
\begin{align*}
g_{lk}\cdot(-1)^{\abs{k}\abs{m}+\abs{k}+\abs{m}}g^{mk}
&=(-1)^{\abs{k}\abs{m}+\abs{k}+\abs{m}+(\abs{l}+\abs{k})(\abs{m}+\abs{k})+\abs{k}\abs{l}}g^{mk}g_{kl}\\
&=(-1)^{\abs{m}+\abs{l}\abs{m}}g^{mk}g_{kl}
=(-1)^{\abs{m}+\abs{l}\abs{m}}\delta^m_l=\delta^m_l
\end{align*}
\end{proof}

By an extension of the Gram-Schmidt procedure as detailed e.g. in Sec. 2.8 of \cite{DeW84},
there is an adapted local basis $(e_1,\ldots,e_{t+s+2n})$ of $\mS M$ with $t+s=m$,
such that $g=g_0$ on the level of matrices, which we shall call an $OSp_{(t,s)|2n}$-frame,
with $g_0$ as follows.
\begin{align}
\label{eqnStandardMetric}
\renewcommand{\arraystretch}{1.5}
g_0&:=\left(\arr{@{\;}c@{\;}|@{\;}c@{\;}}
{G_{t,s}&0\\\hline 0&J_{2n}}\right)\quad\mathrm{where}\\
G_{t,s}&:=\left(\arr{cccc}{-1_{t\times t}&0\\0&1_{s\times s}}\right)\;,\quad
J_{2n}:=\left(\arr{ccc}{J_2&0&0\\0&\ddots&0\\0&0&J_2}\right)\;,\quad
J_2:=\left(\arr{cc}{0&-1\\1&0}\right)\nonumber
\end{align}
Moreover, we set
\begin{align*}
Je_k:=\left\{\arr{ll}{-e_k&k\leq t\\e_k&t<k\leq t+s\\e_{k+1}&k=t+s+2l-1\\-e_{k-1}&k=t+s+2l}\right.
\end{align*}
This is such that $\scal[g]{e_k}{Je_j}=(-1)^{\abs{e_k}}\delta_{kj}$ and, moreover, every $v\in M$ has the expansion
\begin{align}
\label{eqnFrameExpansion}
v=\scal[g]{v}{e_j}Je_j=(-1)^{\abs{e_j}}\scal[g]{v}{Je_j}e_j
\end{align}

In the following, we will assume $M$ to be superoriented in the sense that it has an atlas of coordinate charts
$U_i\cong\bR^{m|2n}$ such that, for every coordinate transformation
$\varphi=(\varphi_0,\varphi^{\sharp}):\bR^{m|2n}\rightarrow\bR^{m|2n}$,
both $\det(d\varphi_0)>0$ and $\sdet(d\varphi)>0$ (cf. Sec. 4 of \cite{Gro13}).
In accord with \cite{AHP12}, sections of the superdeterminant sheaf $\sdet M$
(also known as Berezinian sheaf, see Chp. 3 of \cite{DM99})
will be referred to as Berezinian forms. By definition, a Berezinian form $\omega\in\sdet M$ has
the local form $[dy]\cdot f$ with respect to coordinates $y$ and transforms according to
\begin{align}
\label{eqnBerezinTransformation}
[dy]\cdot f=[dx]\cdot\frac{[d\varphi^*(y)]}{[dx]}\cdot\varphi^*(f)\;,\qquad\frac{[d\varphi^*(y)]}{[dx]}:=\sdet(d\varphi)
\end{align}
where $x$ are different coordinates and $\varphi$ denotes the coordinate transformation.
When the reference to $\varphi$ is clear, we shall also abbreviate (\ref{eqnBerezinTransformation}) as
$[dy]\cdot f=[dx]\cdot\frac{[dy]}{[dx]}\cdot f$.

Let $\varphi:M\rightarrow N$ be an isomorphism of supermanifolds and $\omega\in\sdet N$.
In the following, we will tacitly assume that isomorphisms are orientation preserving. The pullback of $\omega$
under $\varphi$ is defined, locally for $\omega=[dy]\cdot f$, as
\begin{align}
\label{eqnBerezinPullback}
\varphi^*\omega:=[d\varphi^*(y)]\cdot\varphi^*(f)
\end{align}
Note that $\varphi^*(y)$ is a coordinate system. By (\ref{eqnBerezinTransformation}), this construction is well-defined.

We define the integral over a Berezinian form $[dx]\cdot f$ on a coordinate chart $M\supseteq U\cong\bR^{m|2n}$ with coordinates
$x=(u^1,\ldots,u^m,\theta^1,\ldots,\theta^{2n})$ to be the Berezin integral
\begin{align}
\label{eqnBerezinIntegral}
\int_U[dx]\,f:=\int_{\bR^{m|2n}} d^mud^{2n}\theta f:=(-1)^{s(m,2n)}\int_{\bR^m}d^mu\,f^{(1,\ldots,1)}
\end{align}
where $f^{(1,\ldots,1)}$ is the coefficient corresponding to the multiindex $I=(1,\ldots,1)$ in the expansion
$f=\theta^I\cdot f_I$ and $s(m,2n)$ is a sign, which is usually set by convention to
\begin{align}
\label{eqnBerezinSign}
s(m,2n):=n(2n-1)\qquad\textrm{such that}\qquad\int d^{2n}\theta=\dd{\theta^1}\ldots\dd{\theta^{2n}}
\end{align}
By the transformation formula of Berezin integration (Thm. 4.6.1 in \cite{Var04})
and our assumption of classical orientedness ($\det(d\varphi_0)>0$
for coordinate transformations $\varphi$), (\ref{eqnBerezinIntegral}) induces a well-defined
integral $\int_M\omega$ for $\omega\in\sdet M$, provided that $M$ is compact.
While the assumption of orientedness may be easily dropped by considering Berezin densities instead of forms,
compactness of $M$ is essential. The non-compact case was studied in \cite{AHP12}, and will be summarised
below.

By our second condition ($\sdet(d\varphi)>0$), the metric $g$ induces a canonical Berezinian form
$\dsvol_g$, referred to as super volume form, which has the local form
\begin{align}
\label{eqnVolumeForm}
\dsvol_g=[dx]\cdot\sqrt{\abs{\sdet g_x}}
\end{align}
in coordinates $x$ where $\sdet(g_x)$ denotes the super-determinant of the matrix
$g_x:=\scal[g]{\partial_{x^l}}{\partial_{x^k}}$. By (\ref{eqnBerezinPullback}), the pullback of $\dsvol_g$
under an automorphism $\varphi:M\rightarrow M$ reads
\begin{align}
\label{eqnSuperVolumeFormPullback}
\varphi^*\dsvol_g=\dsvol_g\cdot\frac{1}{\sqrt{\abs{\sdet g}}}\sdet(d\varphi)\varphi^*\left(\sqrt{\abs{sdet g}}\right)
\end{align}

We will also need to consider relative Berezinian forms. Letting $S$ be another supermanifold,
we define the sheaf of $S$-Berezinian forms by
\begin{align*}
\sdet_SM:=\sdet M\otimes_{\mO_M}\mO_{M\times S}
\end{align*}
An $S$-form $\omega\in\sdet_SM$ has the local form $\omega=[dx]\cdot f$ with $x$ denoting coordinates on $M$
and $f\in\mO_{M\times S}$. The integral over $\omega$ in (\ref{eqnBerezinIntegral})
is now a function $\int_U\omega\in\mO_S$. This construction is used in the context of maps and vector fields
with flesh where $S=\bR^{0|L}$. The latter terminology was introduced in \cite{Hel09},
while the same concept occurs with several names in the literature, see \cite{DF99b} and \cite{Khe07}.
Maps with flesh allow for having ''odd component fields'' and are deeply related to inner Hom objects
in the category of supermanifolds \cite{SW11}. For details on the derived differential calculus, see \cite{Gro11a},
while an application is given in Sec. \ref{subsecConservedQuantities} below.

For the Lie derivative to be introduced next, we need the pullback of a form $\omega\in\sdet M$ under a morphism
$\varphi:M\times S\rightarrow M$, which we define locally as
\begin{align}
\label{eqnRelativePullback}
\varphi^*([dy]\cdot f):=[d\varphi^*(y)]\cdot\varphi^*(f):=[dx]\cdot\sdet\left(d\varphi|_{\mS M\otimes\mO_{M\times S}}\right)\cdot\varphi^*(f)
\end{align}
where $d\varphi$ is identified with a matrix with respect to coordinates $x$ and $y$ on $M$ considered as domain and range, respectively.
This construction yields a well-defined $S$-form $\varphi^*\omega\in\sdet_SM$, provided that $d\varphi|_{\mS M\otimes\mO_{M\times S}}$ is invertible
and orientation preserving. In this case, the analogon of (\ref{eqnSuperVolumeFormPullback}) continues to hold.

\subsection{The Divergence of a Vector Field}

As shown in \cite{MSV93}, every super vector field $X\in\mS M$ possesses a unique flow, that is
a morphism $\varphi=(\varphi_0,\varphi^{\sharp}):\mD(X)\rightarrow M$ such that
\begin{align}
\label{eqnFlow}
\ev|_{t=t_0}\circ D\circ\varphi^{\sharp}=\ev|_{t=t_0}\circ\varphi^{\sharp}\circ X\;,\qquad
\ev|_{t=0}\circ\varphi^{\sharp}=\id
\end{align}
holds, where $\mD(X)$ is an open subsupermanifold of $\bR^{1|1}\times (M,\mO_M)$ and $D$ can be chosen
to be the lift to $\mD(X)$ of the vector field $\partial_t+\partial_{\tau}$ on $\bR^{1|1}$ with global coordinates $(t,\tau)$.
The pullback of a Berezinian form $\omega\in\sdet M$ under the flow is defined according to (\ref{eqnRelativePullback}).
By the second condition in (\ref{eqnFlow}), the necessary condition on $d\varphi$ is satisfied at least for sufficiently
small $t$ (depending on $x\in M_0$), such that the following definition of the Lie derivative
of a Berezinian $\omega\in\sdet M$ makes sense.
\begin{align*}
L_X\omega:=\ev|_{t=0}D\circ\varphi^*\omega\in\sdet M
\end{align*}
The Lie derivative of tensors can be defined analogous, see Sec. 3 of \cite{Gro13}.

\begin{Def}
\label{defDivergence}
Let $X\in\mS M$ be a super vector field on a semi-Riemannian supermanifold $(M,g)$. The
divergence of $X$ is defined as follows.
\begin{align*}
\dsvol_g\cdot\div X:=L_X\dsvol_g
\end{align*}
\end{Def}

We will show in the following proposition that, as in the classical case,
the divergence of a vector field can be characterised by a local formula, which can as well be taken
as a definition. In the proof of the divergence Thm. \ref{thmSuperDivergence} it will, however,
be helpful to use Def. \ref{defDivergence} directly. For the characterisation, we will need the following two Jacobi lemmas.

\begin{Lem}[Jacobi's Formula]
\label{lemJacobi}
Let $A$ be a supercommutative superalgebra with even part $A_0$.
Let $X\in GL_n(A_0)$ and $\xi$ be a super derivation on $A$. Then
\begin{align*}
\xi(\det X)=\det(X)\cdot\tr\left(X^{-1}\cdot\xi(X)\right)
\end{align*}
\end{Lem}

\begin{proof}
The determinant is polynomial in the entries of $X$ such that the chain rule $\xi(\det X)=\dd[\det X]{X_{kl}}\cdot\xi(X_{kl})$
holds. The statement then follows from the observation that the first factor is the $(lk)$-th entry of $\mathrm{adj}(X)$.
\end{proof}

\begin{Lem}[Jacobi's Formula]
\label{lemSuperJacobi}
Let $A$ be a supercommutative superalgebra.
Let $X\in GL_{m|2n}(A)$ and $\xi$ be an even super derivation on $A$. Then
\begin{align*}
\xi(\sdet X)=\sdet X\cdot\str(X^{-1}(\xi X))
\end{align*}
\end{Lem}

\begin{proof}
First, the statement holds for $X$ of the form
\begin{align*}
\left(\arr{cc}{A&0\\0&D}\right)\;,\qquad\left(\arr{cc}{1&B\\0&1}\right)\;,\qquad\left(\arr{cc}{1&0\\C&1}\right)
\end{align*}
Indeed, the first case follows from applying Lem. \ref{lemJacobi} to the matrices $A$ and $D$ in
\begin{align*}
\xi(\sdet X)=\xi\left(\det A\cdot\det D^{-1}\right)=\xi(\det A)\cdot\det D^{-1}+\det A\cdot\xi(\det D^{-1})
\end{align*}
while a short calculation yields the second and third cases.

For $S$ denoting either side of the proposed equation, one obtains the following multiplicative structure.
\begin{align*}
S(XY)=S(X)\sdet Y+\sdet X\cdot S(Y)
\end{align*}
Since any $X$ can be written as a triple product of matrices of the form stated (cf. P. 118 in \cite{Var04}),
the statement is proved.
\end{proof}

\begin{Prp}
\label{prpDivCharacterisation}
The divergence $\div X$ of a super vector field $X$ is characterised by either of the following formulas
in coordinates $x$ or with respect to a local $OSp$-frame $(e_j)$, respectively.
\begin{enumerate}
\renewcommand{\labelenumi}{(\roman{enumi})}
\item
$\div X=\dd{x^k}(X^k)+\frac{1}{\sqrt{\abs{\sdet g}}}X\left(\sqrt{\abs{\sdet g}}\right)
=\frac{1}{\sqrt{\abs{\sdet g}}}\dd{x^k}\left(\sqrt{\abs{\sdet g}}X^k\right)$
\item
$\div X=\dd{x^k}(X^k)+\frac{1}{2}(-1)^{\abs{m}+\abs{l}(\abs{m}+\abs{k})}g^{mk}\cdot\dd{x^l}g_{km}\cdot X^l$
\item
$\div X=\str\left(Y\mapsto(-1)^{\abs{X}\abs{Y}}\nabla_YX\right)=(-1)^{\abs{e_j}\abs{X}}\scal[g]{\nabla_{e_j}X}{Je_j}$
\end{enumerate}
\end{Prp}

\begin{proof}
$(i)$: In the following, we shall abbreviate $d\varphi|_{\mS M\otimes\mD(X)}$ by $d\varphi$.
Using (\ref{eqnSuperVolumeFormPullback}) in the sense of (\ref{eqnRelativePullback}) and (\ref{eqnFlow}), we calculate
\begin{align*}
\div X&=\ev|_{t=0}D\left(\frac{1}{\sqrt{\abs{\sdet g}}}\sdet(d\varphi)\varphi^*\left(\sqrt{\abs{sdet g}}\right)\right)\\
&=\ev|_{t=0}D\,\sdet(d\varphi)+\frac{1}{\sqrt{\abs{\sdet g}}}X\left(\sqrt{\abs{sdet g}}\right)
\end{align*}
We calculate the first term. Assume $X$ is homogeneous. If $X$ is odd, let $\eta$ be an additional Grassmann generator
(i.e. replace the supermanifold $M$ by $M\times\bR^{0|1}$). If $X$ is even, we let $\eta:=1$.
In either case, $\eta\cdot D$ is even such that Lem. \ref{lemSuperJacobi} is applicable.
With $\ev_{t=0}d\varphi=\id$, we thus yield
\begin{align*}
\eta\cdot\ev|_{t=0}D\sdet(d\varphi)
=\ev|_{t=0}\left(\sdet d\varphi\cdot\str(d\varphi^{-1}(\eta D)d\varphi)\right)
=\eta\cdot\ev|_{t=0}\str(D\,d\varphi)
\end{align*}
Comparing coefficients and using (\ref{eqnVectorField}), we further calculate
\begin{align*}
\ev|_{t=0}D\,\sdet(d\varphi)
=(-1)^{\abs{x^k}}\ev|_{t=0}D\dd{x^k}\varphi^*(x^k)
=(-1)^{\abs{x^k}(1+\abs{X})}\dd{x^k}\ev|_{t=0}D\varphi^*(x^k)
\end{align*}
With (\ref{eqnFlow}), this calculation completes the proof of part $(i)$.

$(ii)$: With part $(i)$ and (\ref{eqnCalculationRule}), we find
\begin{align*}
\div X=\partial_{x^k}(X^k)+\frac{1}{2}\frac{1}{\abs{\sdet g}}X(\abs{\sdet g})=\partial_{x^k}(X^k)+\frac{1}{2}\frac{1}{\sdet g}X(\sdet g)
\end{align*}
Part $(ii)$ then follows by introducing $\eta$ as in the proof of part $(i)$ and again using Lem. \ref{lemSuperJacobi}.

$(iii)$: It is clear that the $OSp$-frame expression stated equals the super trace expression, see \cite{Gro13}.
The statement is thus proved by showing that the super trace expression yields the same local result as
gained for $\div X$ in part $(ii)$. Abbreviating $\partial_k:=\partial_{x^k}$, we have
\begin{align*}
\nabla_{\partial_k}X=\partial_m\cdot\left((-1)^{\abs{k}\abs{m}}\partial_k(X^m)+(-1)^{\abs{m}(1+\abs{k}+\abs{l})}\Gamma_{kl}^m\cdot X^l\right)
\end{align*}
and thus
\begin{align*}
&\str(Y\mapsto(-1)^{\abs{X}\abs{Y}}\nabla_YX)\\
&\qquad=(-1)^{\abs{k}(\abs{X}+1)}(-1)^{\abs{X}\abs{k}}\left((-1)^{\abs{k}}\partial_k(X^k)+(-1)^{\abs{k}(1+\abs{k}+\abs{l})}\Gamma_{kl}^k\cdot X^l\right)\\
&\qquad=\partial_k(X^k)+(-1)^{\abs{k}(1+\abs{l})}\Gamma_{kl}^k\cdot X^l
\end{align*}
We express the Christoffel symbols as in Lem. \ref{lemChristoffel} to rewrite the second term as
\begin{align*}
\frac{1}{2}(-1)^{\abs{k}(1+\abs{l})}\left(\partial_kg_{lm}+(-1)^{\abs{k}\abs{l}}\partial_lg_{km}-(-1)^{\abs{m}(\abs{k}+\abs{l})}\partial_mg_{kl}\right)g^{mk}X^l
\end{align*}
By a straightforward calculation, using $g^{mk}=(-1)^{\abs{m}\abs{k}+\abs{m}+\abs{k}}g^{km}$ (Lem. \ref{lemChristoffel})
and renaming indices, we see that the first and third summands together vanish,
from which part $(iii)$ follows.
\end{proof}

As a direct consequence of Prp. \ref{prpDivCharacterisation}$(iii)$, we find

\begin{Lem}
\label{lemDivergenceRule}
The divergence satisfies
\begin{align*}
\div(fX)=f\cdot\div X+(-1)^{\abs{f}\abs{X}}X(f)
\end{align*}
\end{Lem}

Moreover, we introduce the gradient $\nabla f$ and laplacian $\laplace f$ of a function $f\in\mO_M$ via
\begin{align}
\label{eqnDefGradientLaplace}
\scal[g]{\nabla f}{Y}:=df[Y]=(-1)^{\abs{Y}\abs{f}}Y(f)\;,\qquad\laplace f:=-\div(\nabla f)
\end{align}
The gradient has the following local formula.
\begin{align}
\label{eqnGradientfg}
\nabla f=(-1)^{\abs{e_j}\abs{f}}e_j(f)Je_j=(-1)^{\abs{e_j}+\abs{e_j}\abs{f}}Je_j(f)\cdot e_j
\end{align}
and a straightforward calculation yields
\begin{align}
\label{eqnGradientDifference}
(\nabla f)(k)-(-1)^{\abs{f}\abs{k}}(\nabla k)(f)=0
\end{align}
Inserting (\ref{eqnGradientfg}) into Prp. \ref{prpDivCharacterisation}$(iii)$, we yield the local formula
\begin{align}
\label{eqnLaplaceLocal}
\laplace f=(\nabla_{e_j}Je_j)(f)-e_j\circ Je_j(f)
\end{align}
Moreover, Lem. \ref{lemDivergenceRule} and (\ref{eqnGradientfg}) and (\ref{eqnGradientDifference}) immediately imply
\begin{align}
\label{eqnLaplaceDiv}
f\laplace k&=-\div(f\nabla k)+(-1)^{\abs{f}\abs{k}}(-1)^{\abs{e_j}\abs{k}}e_j(k)Je_j(f)\\
\label{eqnLaplaceDifference}
f\laplace k-(-1)^{\abs{f}\abs{k}}k\laplace f&=-\div(f\nabla k)+(-1)^{\abs{f}\abs{k}}\div(k\nabla f)
\end{align}

We will also need the divergence of a vector field along a supermanifold morphism $\varphi:(M,h)\rightarrow(X,g)$
between semi-Riemannian supermanifolds, which is allowed to be a map with flesh.
Following \cite{Han12} and \cite{Gro13}, $\varphi$ is called superharmonic if its tension field
\begin{align}
\label{eqnTensionField}
\tau(\varphi):=\str_h\left((X,Y)\mapsto(\nabla_Xd\varphi)[Y]\right)=(\nabla_{e_j}d\varphi)[Je_j]
\end{align}
vanishes.

\begin{Def}
Let $\xi\in\mS\varphi$ be a vector field along $\varphi$.
We define its divergence to be
\begin{align*}
\div\xi:=\str_h\left((X,Y)\mapsto(-1)^{\abs{X}\abs{\xi}}\scal[g_{\varphi}]{\nabla_{X}\xi}{d\varphi[Y]}\right)
=(-1)^{\abs{e_i}\abs{\xi}}\scal[g_{\varphi}]{\nabla_{e_i}\xi}{d\varphi[Je_i]}
\end{align*}
\end{Def}

\begin{Lem}[Lem. 4.7 of \cite{Gro13}]
\label{lemDivAlongMorphism}
Let $\xi\in\mS\varphi$. Then
\begin{align*}
\div\xi=-\scal[g_{\varphi}]{\xi}{\tau(\varphi)}+\div\left(\scal[g_{\varphi}]{\xi}{d\varphi[e_j]}Je_j\right)
\end{align*}
\end{Lem}

\subsection{Integration on Supermanifolds with Boundary}

On a non-compact supermanifold, the transformation formula for Berezin integration holds only up
to boundary terms, such that the local integrals (\ref{eqnBerezinIntegral}) no longer glue to a globally
well-defined expression. The nature of those boundary terms has been analysed in \cite{AHP12}
and will now be summarised. It turns out that a global integral can be defined
by introducing a retraction as an additional datum.

\begin{Def}
A morphism $\gamma=(\gamma_0,\gamma^{\sharp}):M\rightarrow M_0$ is called a retraction if $\gamma\circ j_M=\id_{M_0}$ with
$j_M:M_0\rightarrow M$ the canonical embedding.
\end{Def}

Let $x=(u,\theta)$ be coordinates in $U\subseteq M$ and $u_0$ the underlying classical coordinates on $U_0$.
Then there is a unique retraction $\gamma$ on $U$ defined by $\gamma^*(u_0)=u$, with respect to which every superfunction
$f\in\mO_U$ possesses a unique decomposition
\begin{align*}
f=\sum_I\theta^I\cdot\gamma^*(f_I)
\end{align*}
and one can define a global map $\gamma_!:\sdet M\rightarrow\det M_0=\Omega^m M_0$ by the local expression
\begin{align*}
(\gamma_!([dx]\,f))|_U:=(-1)^{s(m,2n)}\abs{du_0}\,f^{(1,\ldots,1)}
\end{align*}
with $s(m,2n)$ as in (\ref{eqnBerezinIntegral}).

\begin{Def}
\label{defRetractionIntegral}
Let $\gamma$ be a retraction on $M$ and $\omega$ be a Berezin density. $\omega$ is called integrable with respect to $\gamma$
if $\gamma_!(\omega)$ is integrable on $M_0$. In this case, we define
\begin{align*}
\int_{(M,\gamma)}\omega:=\int_{M_0}\gamma_!\omega
\end{align*}
\end{Def}

The integral over $M=\bR^{m|2n}$ with respect to the standard retraction $\gamma_{\mathrm{std}}$, defined by
$\gamma_{\mathrm{std}}^*(x^j):=x^j$, of a Berezinian form $\omega=[dx]\,f$ with $f\in\mO_M$ then becomes the
standard Berezin integral
\begin{align*}
\int_{(M,\gamma_{\mathrm{std}})}[dx]\,f=\int_{\bR^{m|2n}}d^mxd^{2n}\theta\, f=(-1)^{s(m,2n)}\int_{\bR^m}d^mx\,f^{1,\ldots,1}(x)
\end{align*}
Only if $\omega$ is compactly supported, the integral in Def. \ref{defRetractionIntegral} is independent of the choice of $\gamma$,
see Thm. \ref{thmChangeOfRetraction} below.
Let $\varphi:N\rightarrow M$ be an orientation preserving isomorphism of supermanifolds and define the pullback
of a retraction $\gamma$ by $\varphi^*\gamma:=\varphi_0^{-1}\circ\gamma\circ\varphi:N\rightarrow N_0$.

\begin{Lem}[Cor. 2.15 in \cite{AHP12}]
\label{lemPullbackIntegral}
Let $\omega\in\sdet M$ be integrable with respect to $\gamma$. Then $\varphi^*\omega$ is integrable with
respect to $\varphi^*\gamma$ and
\begin{align*}
\int_{(N,\varphi^*\gamma)}\varphi^*\omega=\int_{(M,\gamma)}\omega
\end{align*}
\end{Lem}

We are interested in integration over an open subsupermanifold $N\subseteq M$ of $M$ with smooth boundary.
More specifically, let $N$ be such that $N_0$ is defined
as the open subset of points $p\in M_0$ such that $\rho(p)<C$ with $C\in\bR$ a constant for a \emph{boundary function} $\rho:M_0\rightarrow\bR$
with full rank Jacobian at all its points which map to $C$. The boundary $\partial N_0$ is then the set of $p\in M$
such that $\rho(p)=C$. Now let $\tau\in\mO_M$ be an even superfunction such that the underlying function $\rho=\tau_0$
is a boundary function. By Prp. 3.2.6 in \cite{Lei80}, there exists a supermanifold $\partial_{\tau}N$ of dimension $m-1|2n$
along with an immersion
$\iota_{\tau}:\partial_{\tau}N\rightarrow M$ such that
\begin{align*}
(\partial_{\tau}N)_0=\partial N_0\;,\qquad(\iota_{\tau})_0=\iota_0\;,\qquad\iota_{\tau}^*(\tau)=C
\end{align*}
where $\iota_0$ is the inclusion $\iota_0:\partial N_0\rightarrow M_0$.
The pair $(\partial_{\tau}N,\iota_{\tau})$ is uniquely determined up to equivalence. The local picture is as follows. In a neighbourhood of a point $p\in M_0$,
$\tau$ can be endowed to a coordinate system $x=(\tau,\tilde{x})$ of $M$. In such coordinates, $\partial_{\tau}N$ is obtained
by setting the boundary coordinate $\tau$ to $C$, and the remaining coordinate vector fields locally span
its tangent sheaf. More specifically, let $\varphi=(\varphi_0,\varphi^{\sharp}):\bR^{m|2n}\rightarrow U$ denote the (inverse) morphism
corresponding to coordinates $x=(\tau,\tilde{x})$. Then $\partial_{\tau}N\cap U\cong\bR^{m-1|2n}$, and
$\mS \partial_{\tau}N$ is locally identified with a subsheaf of $\mS M$ as follows.
\begin{align}
\label{eqnBoundaryVectorFields}
\mS\partial_{\tau}N(\partial N_0\cap U_0)
\cong(\varphi^{\sharp})^{-1}\circ\spann_{\mO_{\bR^{m-1|2n}}}\left(\dd{\tilde{x}^1},\ldots,\dd{\tilde{x}^{m+2n-1}}\right)\circ\varphi^{\sharp}
\subseteq\mS M(U_0)
\end{align}
Moreover, there is a restriction map $\omega\mapsto\omega|_{\partial_{\tau}N,\tau}$ which is
defined locally in coordinates $x=(\tau,\tilde{x})$ by setting
\begin{align}
\label{eqnRestrictionMap}
([dx]\,f)|_{\partial_{\tau}N,\tau}:=[d\iota_{\tau}^*(\tilde{x})]\,\iota_{\tau}^*(f)
\end{align}
This yields a well-defined map $\sdet M\rightarrow\sdet\partial_{\tau}N$.

Now let $\gamma$ be a retraction on $M$ and $\rho$ be a classical boundary function. Denoting
the boundary supermanifold of $\tau=\gamma^*(\rho)$ by $\partial_{\gamma}N:=\partial_{\tau}N$
with immersion $\iota_{\gamma}:=\iota_{\tau}$, it was shown in \cite{AHP12} that there exists a unique
retraction $\partial\gamma$ on $\partial_{\gamma}N$ such that
\begin{align}
\label{eqnBoundaryDiagram}
\begin{xy}
\xymatrix{\ar@{}[dr]|{\circlearrowleft}
\partial_{\gamma}N\ar[d]_{\partial\gamma}\ar[r]^{\iota_{\gamma}}&M\ar[d]^{\gamma}\\
\partial N_0\ar[r]_{\iota_0}&M_0}
\end{xy}
\end{align}
commutes, and such that (\ref{eqnRestrictionMap}) is compatible with $\gamma_!$ and $(\partial\gamma)_!$.
By a direct calculation, the boundary integral
\begin{align}
\label{eqnBoundaryIntegral}
\int_{(\partial_{\gamma}N,\partial\gamma)}\omega|_{\partial_{\gamma}N,\tau}
\end{align}
is independent of the representative $(\partial_{\gamma}N,\iota_{\gamma})$ in the aforementioned equivalence class.
As shown in \cite{Che94}, there is a natural action of differential operators on Berezinian forms which,
locally for operators of the form $(\partial_{x^j})^k$ with respect to coordinates $x$, takes the form
\begin{align*}
([dx]\,f).(\partial_{x^j})^k=[dx]\,(-1)^k(\partial_{x^j})^k(f)
\end{align*}
With this notation established, we can now state one of the main results of \cite{AHP12}.

\begin{Thm}[Cor. 5.19 in \cite{AHP12}]
\label{thmChangeOfRetraction}
Let $N\subseteq M$ be open with smooth boundary $\partial N_0$, and let $\gamma,\gamma'$ be retractions
on $M$. Then, for compactly supported $\omega\in\sdet M$,
\begin{align*}
\int_{(N,\gamma')}\omega&=\int_{(N,\gamma)}\omega+b_{N,\gamma',\gamma}(\omega)
\end{align*}
with the boundary term given by
\begin{align*}
b_{N,\gamma',\gamma}(\omega)&=-S\cdot\sum_{j=1}^n\frac{1}{j!}\int_{(\partial_{\gamma}N,\partial\gamma)}((((\gamma'^*(\rho)-\gamma^*(\rho))^j\omega).\partial_{\tau}^{j-1})|_{\partial_{\gamma}N,\gamma^*(\rho)}
\end{align*}
where $S:=(-1)^{s(m,2n)+s(m-1,2n)}$.
\end{Thm}

\section{Divergence Theorems in Supergeometry}
\label{secDivergenceTheorems}

In this section, we shall establish divergence theorems for semi-Riemannian supermanifolds
by means of the flow (\ref{eqnFlow}) of a vector field and Thm. \ref{thmChangeOfRetraction},
concerning the change of retractions, and finally apply the results obtained to the study
of conserved quantities in the context of superharmonic functions.
As in classical semi-Riemannian geometry, one has to take care about degeneracy of the boundary metric.
While our focus is on the non-degenerate case, we will also yield a generalisation of \"Unal's divergence theorem
for a degenerate boundary.

As before, let $(M,g)$ be a semi-Riemannian supermanifold. We define the boundary metric
on a subsupermanifold $(\partial_{\gamma}N,\iota_{\gamma})$ as in (\ref{eqnBoundaryDiagram})
to be the pullback $\tilde{g}:=\iota_{\gamma}^*g$ in (\ref{eqnPullback}).
In general, an even supersymmetric bilinear form $g$ on a supermanifold $M$ is non-degenerate if and only if the induced form
$g_p$ on the super tangent space $T_pM$ is non-degenerate for all $p\in M_0$.
This is shown e.g. by Lem. 3.38 of \cite{Han12}.
For the pullback $\tilde{g}$ on the boundary supermanifold $\partial_{\gamma}N$, this implies that
non-degeneracy is solely determined by non-degeneracy of $\tilde{g}_p$ restricted to the even part
of the super vector space $T_p\partial_{\gamma}N$. We thus yield the following criterion.

\begin{Lem}
\label{lemNonDegeneracy}
$\tilde{g}$ is non-degenerate if and only if $g_0|_{\partial N_0}$ (the underlying classical metric $g_0$ on $M_0$ restricted
to the boundary $\partial N_0$) is non-degenerate.
\end{Lem}

Unless stated otherwise, we will assume that $\tilde{g}$ is non-degenerate. This is in particular the case if $g$
is Riemannian (which, by definition, means that $g_0$ is Riemannian). By assumption, there
is a unit normal vector field near the boundary, i.e. a vector field $\nu$ such that $\nu\perp_g\mS\partial_{\gamma}N$
and $\varepsilon(\nu):=\scal[g]{\nu}{\nu}\equiv\pm 1$. $\nu$ is unique up to a sign which we fix such that $\nu_0$
is outward-pointing.

\begin{Lem}
\label{lemInducedMetric}
Endow $\tau$ to a coordinate system $x=(\tau,\tilde{x})$ in a neighbourhood of a point $p\in\partial N_0\subseteq M_0$. Then
\begin{align*}
\sqrt{\abs{\sdet g_x}}=\varepsilon(\nu)\cdot\scal[g]{\partial_{\tau}}{\nu}\cdot\sqrt{\abs{\sdet\tilde{g}_{\tilde{x}}}}\;,\qquad
\dsvol_g|_{\partial_{\gamma}N,\tau}=\dsvol_{\tilde{g}}\cdot\iota_{\gamma}^*\left(\varepsilon(\nu)\cdot\scal[g]{\partial_{\tau}}{\nu}\right)
\end{align*}
where $g_x$ denotes the matrix $\scal[g]{\partial_{x^i}}{\partial_{x^j}}$.
\end{Lem}

\begin{proof}
Consider the local frame $(f_1,\ldots,f_{m+2n}):=(\nu,\partial_{x^2},\ldots,\partial_{x^{m+2n}})$.
The transition matrix $A$ defined by
$\partial_{x^i}=f_l\cdot A_{li}$ and the metric $g_f$ in the new basis, respectively, then have the form
\begin{align*}
A=\left(\arr{cc}{\varepsilon(\nu)\cdot\scal[g]{\partial_{\tau}}{\nu}&0\\ *&\mathbbm{1}_{(m-1|2n)^2}}\right)\;,\qquad
g_f=\left(\arr{cc}{\varepsilon(\nu)&0\\0&\tilde{g}_{\tilde{x}}}\right)
\end{align*}
such that
\begin{align*}
\sdet g_x=\sdet(A^{ST}g_fA)=\sdet A^2\cdot\sdet g_f=\varepsilon(\nu)\cdot\scal[g]{\partial_{\tau}}{\nu}^2\cdot\sdet\tilde{g}_{\tilde{x}}
\end{align*}
Since $\nu_0$ is outward-pointing, the orientation classes of the two local frames coincide such that
$\varepsilon(\nu)\scal[g]{\partial_{\tau}}{\nu}>0$. This yields the first statement, while the second is immediate by the
first together with (\ref{eqnVolumeForm}) and (\ref{eqnRestrictionMap}).
\end{proof}

\begin{Thm}[Divergence Theorem]
\label{thmSuperDivergence}
Let $N\subseteq M$ be an open submanifold such that $\overline{N_0}$ is compact
and has a smooth boundary $\partial N_0$ defined by a boundary function $\rho$.
Let $\gamma$ be a retraction on $M$, and let $(\partial_{\gamma}N,\iota_{\gamma})$
and $\partial\gamma$ be as in (\ref{eqnBoundaryDiagram}). If $\tilde{g}=\iota_{\gamma}^*g$ is non-degenerate then, for every vector field $X\in\mS M$,
\begin{align*}
\int_{(N,\gamma)}\dsvol_g\,\div X
=S\cdot\varepsilon(\nu)\int_{(\partial_{\gamma}N,\partial\gamma)}\dsvol_{\tilde{g}}\,\iota^*_{\gamma}\left(\scal[g]{X}{\nu}\right)
\end{align*}
where $S:=(-1)^{s(m,2n)+s(m-1,2n)}$.
In particular, the right hand side vanishes if $X$ has compact support in $N$.
\end{Thm}

\begin{proof}
Both sides are $\bR$-linear in $X$, so $X$ may be assumed to be homogeneous. Consider the case of even $X$ first,
and let $\varphi$ denote its flow. By compactness of $\overline{N_0}$,
there is an open submanifold $V\subseteq\bR^{1|1}$ with $0\in V_0$ such that
$V\times N\subseteq\mD(X)$ and $\overline{V_0}$ is compact.
We may thus exchange integration with evaluation and differentiation in the
following calculation, after employing Def. \ref{defDivergence}.
\begin{align*}
\int_{(N,\gamma)}\dsvol_g\,\div X&=\int_{(N,\gamma)}\ev|_{t=0}D\varphi^*\dsvol_g\\
&=\ev|_{t=0}D\int_{(N,\gamma)}\varphi^*\dsvol_g\\
&=\ev|_{t=0}D\int_{(N,\varphi^*\gamma)}\varphi^*\dsvol_g-\ev|_{t=0}Db_{N,\varphi^*\gamma,\gamma}(\varphi^*\dsvol_g)
\end{align*}
The last equation uses Thm. \ref{thmChangeOfRetraction} in the following sense.
By the first equation of (\ref{eqnFlow}) and evenness of $X$,
$\varphi$ is independent of the odd time variable $\tau$ and as such can be considered as a $t$-dependent morphism
$\varphi_t:M\rightarrow M$ which, by the second equation in (\ref{eqnFlow}), is (for fixed $t$) an isomorphism provided that
$V$ is chosen sufficiently small. In this context, $\varphi^*\gamma$ is a retraction on $M$ which also depends on $t$.
By Lem. \ref{lemPullbackIntegral}, the integral in the first term is independent of the flow parameters
and thus vanishes upon differentiation.
If $X$ has compact support in $N$, the second expression vanishes, too.
If not, it remains to calculate the boundary term by use of Lem. \ref{lemInducedMetric} as follows.
We write $\varphi^*\dsvol_g=:\dsvol_g\cdot d_{\varphi}$,
where $d_{\varphi}$ is the superfunction determined by (\ref{eqnSuperVolumeFormPullback}).
\begin{align*}
&\int_{(N,\gamma)}\dsvol_g\,\div X\\
&\qquad=S\cdot\ev|_{t=0}D\,\sum_{j=1}\frac{1}{j!}\int_{(\partial_{\gamma}N,\partial\gamma)}((\dsvol_g\,d_{\varphi}\cdot((\varphi^*\gamma)^*(\rho)-\gamma^*(\rho))^j).\partial_{\tau}^{j-1})|_{\partial_{\gamma}N,\tau}\\
&\qquad=S\cdot\ev|_{t=0}D\,\sum_{j=1}\frac{(-1)^{j-1}}{j!}\int_{(\partial_{\gamma}N,\partial\gamma)}[d\iota_{\gamma}^*(\tilde{x})]\\
&\qquad\qquad\qquad\qquad\qquad\qquad\qquad\qquad\iota_{\gamma}^*\left(\partial_{\tau}^{j-1}\left(\sqrt{\abs{\sdet g}}\cdot d_{\varphi}\cdot((\varphi^*\gamma)^*(\rho)-\gamma^*(\rho))^j\right)\right)\\
&\qquad=S\cdot\varepsilon(\nu)\cdot\ev|_{t=0}\sum_{j=1}\frac{(-1)^{j-1}}{j!}\int_{(\partial_{\gamma}N,\partial\gamma)}\dsvol_{\tilde{g}}\\
&\qquad\qquad\qquad\qquad\qquad\qquad\qquad\qquad\iota_{\gamma}^*\left(\partial_{\tau}^{j-1}\left(\scal[g]{\partial_{\tau}}{\nu}\cdot D(d_{\varphi}\cdot((\varphi^*\gamma)^*(\rho)-\gamma^*(\rho))^j)\right)\right)
\end{align*}
Now the expression
\begin{align*}
&D\left(d_{\varphi}\cdot((\varphi^*\gamma)^*(\rho)-\gamma^*(\rho))^j\right)\\
&\qquad=(Dd_{\varphi})\cdot((\varphi^*\gamma)^*(\rho)-\gamma^*(\rho))^j
+jd_{\varphi}\cdot((\varphi^*\gamma)^*(\rho)-\gamma^*(\rho))^{j-1}D(\varphi^*\gamma)^*(\rho)
\end{align*}
vanishes for $j>1$ upon evaluation, and we are left with
\begin{align*}
\ev|_{t=0}D\circ(\varphi^*\gamma)^*(\rho)&=\ev|_{t=0}D\circ\varphi_t^*\circ\gamma^*\circ(\varphi_{0t}^{-1})^*(\rho)\\
&=\ev|_{s=0}\ev|_{t=0}D\circ\varphi_t^*\circ\gamma^*\circ(\varphi_{0s}^{-1})^*(\rho)\\
&=\ev|_{s=0}X\circ\gamma^*\circ(\varphi_{0s}^{-1})^*(\rho)\\
&=X(\tau)
\end{align*}
using that the expression in the second equation is continuous in $(t,s)$ as well as (\ref{eqnFlow}).
We thus obtain
\begin{align*}
\int_{(N,\gamma)}\dsvol_g\,\div X
=S\cdot\varepsilon(\nu)\int_{(\partial_{\gamma}N,\partial_{\gamma})}\dsvol_{\tilde{g}}\,\iota_{\gamma}^*\left(\scal[g]{\partial_{\tau}}{\nu}X(\tau)\right)
\end{align*}
Since $\nu$ is orthogonal to $\mS\partial_{\gamma}N$, this proves the theorem for even $X$.

Finally, consider an odd vector field $X$. As in the proof of Prp. \ref{prpDivCharacterisation},
we endow $M$ to $M\times\bR^{0|1}$ by introducing an additional odd coordinate $\eta$ and consider the
even vector field $\eta\cdot X$ with its flow $\varphi$. Analogous to above, we find
\begin{align*}
&\int_{\bR^{0|1}}d\eta\int_{(N,\gamma)}\dsvol_g\,\div(\eta X)\\
&\qquad\qquad=\ev|_{t=0}D\int_{\bR^{0|1}}d\eta\int_{(N,\varphi^*\gamma)}\varphi^*\dsvol_g
-\ev|_{t=0}D\int_{\bR^{0|1}}d\eta\, b_{N,\varphi^*\gamma,\gamma}(\varphi^*\dsvol_g)
\end{align*}
Again, the first term vanishes, and the second is calculated as before, thus yielding an analogous statement
for $\eta X$ upon integration over $\eta$ on both sides. From this, the statement for $X$ directly follows.
\end{proof}

The right hand side integral in Thm. \ref{thmSuperDivergence} depends on the boundary supermanifold
$(\partial_{\gamma}N,\iota_{\gamma})$, but not on the induced retraction $\partial\gamma$ since $\partial N_0$ is by assumption compact. If the retraction on $M$ is changed, one obtains an additional boundary term due to Thm. \ref{thmChangeOfRetraction} as follows.

\begin{Cor}
\label{corSuperDivergence}
In the situation of Thm. \ref{thmSuperDivergence}, let $\gamma':M\rightarrow M_0$
be another retraction. Then
\begin{align*}
\int_{(N,\gamma')}\dsvol_g\,\div X
=S\cdot\varepsilon(\nu)\int_{(\partial_{\gamma}N,\partial\gamma)}\dsvol_{\tilde{g}}\,\iota^*_{\gamma}\left(\scal[g]{X}{\nu}\right)
+b_{N,\gamma',\gamma}(\dsvol_g\,\div X)
\end{align*}
with
\begin{align*}
&b_{N,\gamma',\gamma}(\dsvol_g\,\div X)\\
&\;\;\;=-S\cdot\varepsilon(\nu)\sum_{j=1}\frac{(-1)^{j-1}}{j!}\int_{(\partial_{\gamma}N,\partial\gamma)}\dsvol_{\tilde{g}}\,\iota_{\gamma}^*\left(\partial_{\tau}^{j-1}(\scal[g]{\partial_{\tau}}{\nu}\div X(\gamma'^*(\rho)-\gamma^*(\rho))^j)\right)
\end{align*}
\end{Cor}

As a second corollary of Thm. \ref{thmSuperDivergence}, we state the supergeometric version of Green's formula next, which is immediate by (\ref{eqnLaplaceDifference}).

\begin{Cor}[Green's Formula]
In the situation of Thm. \ref{thmSuperDivergence}, let $f,k\in\mO_M$. Then
\begin{align*}
&\int_{(N,\gamma)}\dsvol_g\left(f\laplace k-(-1)^{\abs{f}\abs{k}}k\laplace f\right)\\
&\qquad=-S\cdot\varepsilon(\nu)\int_{(\partial_{\gamma}N,\partial\gamma)}\dsvol_g\,\iota_{\gamma}^*\left(f\cdot\nu(k)-(-1)^{\abs{f}\abs{k}}k\cdot\nu(f)\right)
\end{align*}
\end{Cor}

We next state a divergence theorem for vector fields along a morphism $\varphi$ of semi-Riemannian supermanifolds.
Here, $\tau(\varphi)$ denotes the tension field as defined in (\ref{eqnTensionField}).

\begin{Thm}[Divergence Theorem]
\label{thmMorphismSuperDivergence}
Let $\varphi:(M,h)\rightarrow(X,g)$ be a supermanifold morphism between semi-Riemannian supermanifolds
and $\xi\in\mS\varphi$ be a vector field  along $\varphi$.
Then, in the situation of Thm. \ref{thmSuperDivergence} (with $\tilde{g}$ replaced by $\tilde{h}:=\iota_{\gamma}^*h$),
\begin{align*}
&\int_{(N,\gamma)}\dsvol_g\,\div\xi\\
&\qquad\qquad=-\int_{(N,\gamma)}\dsvol_g\,\scal[g_{\varphi}]{\xi}{\tau(\varphi)}
+S\cdot\varepsilon(\nu)\int_{(\partial_{\gamma}N,\partial\gamma)}\dsvol_{\tilde{g}}\,\iota_{\gamma}^*\left(\scal[g_{\varphi}]{\xi}{d\varphi[\nu]}\right)
\end{align*}
In particular, the first term on the right hand side vanishes if $\varphi$ is superharmonic.
\end{Thm}

\begin{proof}
Expressing $\div\xi$ as in Lem. \ref{lemDivAlongMorphism}, the statement follows from Thm. \ref{thmSuperDivergence} with (\ref{eqnFrameExpansion}).
\end{proof}

We now come to the case where $\tilde{g}$ is degenerate.
Following \cite{Uen95}, we denote by $\partial (N_0)_+$ the subset
of the classical boundary where $g_0|_{\partial N_0}$ is spacelike, similarly by $\partial(N_0)_-$ that
for timelike and by $\partial(N_0)_0$ that for lightlike. The first two are open,
and we set $\partial_{\gamma}N_{\pm}:=\partial_{\gamma}N|_{\partial (N_0)_{\pm}}$.

\begin{Thm}[Divergence Theorem]
\label{thmUenalAnalogon1}
Assume that $\tilde{g}$ is degenerate and $\partial (N_0)_0$ has measure zero in $\partial N_0$. Then,
under the remaining hypotheses of Thm. \ref{thmSuperDivergence},
\begin{align*}
\int_{(N,\gamma)}\dsvol_g\,\div X=S\int_{(\partial_{\gamma}N_+,\partial\gamma)}\dsvol_{\tilde{g}}\,\iota_{\gamma}^*\left(\scal[g]{X}{\nu}\right)
-S\int_{(\partial_{\gamma}N_-,\partial\gamma)}\dsvol_{\tilde{g}}\,\iota_{\gamma}^*\left(\scal[g]{X}{\nu}\right)
\end{align*}
\end{Thm}

\begin{proof}
As in the proof of Thm. \ref{thmSuperDivergence}, we find that, for even $X$,
\begin{align*}
\int_{(N,\gamma)}\dsvol_g\,\div X=-\ev|_{t=0}Db_{N,\varphi^*\gamma,\gamma}(\varphi^*\dsvol_g)
\end{align*}
with $\varphi$ denoting the flow of $X$, and analogous for odd $X$.
Since $\partial(N_0)_0$ is assumed to be a nullset, the boundary integral is just the sum of integrals over
$\partial_{\gamma}N_{\pm}$. By Lem. \ref{lemNonDegeneracy}, $\tilde{g}$ restricted to $\partial_{\gamma}N_{\pm}$
is non-degenerate. Both integrals can thus be calculated as in the proof of Thm. \ref{thmSuperDivergence}, with
different signs $\varepsilon(\nu)=\pm 1$ on $\partial_{\gamma}N_{\pm}$, respectively.
\end{proof}

\subsection{Conserved Quantities and Superharmonic Functions and Maps}
\label{subsecConservedQuantities}

Given a superharmonic map, it was shown in Thms. 4.10, 4.12 and 4.15 of \cite{Gro13}
that every Killing vector field induces a divergence-free vector field which,
by our next result, is a conserved quantity.
In the rest of this section, we will then show that
superharmonic functions, to be studied further in Sec. \ref{secMVT} on the supersphere,
are special instances of superharmonic maps and translate Thm. 4.15 of \cite{Gro13} into this situation.

\begin{Prp}
\label{prpConservedQuantity}
Let $\tilde{g}$ be non-degenerate.
Let $X$ be a vector field such that $\div X=0$. Then the integral
\begin{align*}
\int_{\partial N}\dsvol_{\tilde{g}}\,\iota^*\left(\scal[g]{X}{\nu}\right)=0
\end{align*}
is independent of a retraction and vanishes for any boundary supermanifold $\iota:\partial N\rightarrow M$
of even codimension. In this sense, $X$ is a conserved quantity.
\end{Prp}

\begin{proof}
By assumption, the boundary integral in Thm. \ref{thmSuperDivergence} vanishes for every retraction $\gamma$ on $M$.
It is independent of the boundary retraction $\partial\gamma$ by the comment preceding Cor. \ref{corSuperDivergence}.
\end{proof}

Let $\varphi:M\rightarrow\bR$ be a map with flesh, that is a supermanifold morphism $\varphi:M\times\bR^{0|L}\rightarrow\bR$.
Denoting the global coordinate on $\bR$ by $x$, we get a superfunction $f:=\varphi^*(x)\in(\mO_{M\times\bR^{0|L}})_0$.
Conversely, every superfunction $f\in(\mO_{M\times\bR^{0|L}})_0$ defines a map with flesh via $\varphi^*(x):=f$,
such that both notions are equivalent.

\begin{Prp}
\label{prpSuperharmonic}
Let $\varphi:M\rightarrow\bR$ be a map with flesh and $f$ its associated superfunction as above. Then
\begin{align*}
\tau(\varphi)=-\laplace f\cdot(\varphi^{\sharp}\circ\partial_x)
\end{align*}
In particular, $\varphi$ is superharmonic if and only if $f$ is superharmonic.
\end{Prp}

\begin{proof}
Let $X$ be a vector field on $M$. Unwinding the definitions, we find
\begin{align*}
d\varphi[X]=X(f)\cdot(\varphi^{\sharp}\circ\partial_x)
\end{align*}
For the pullback connection, this implies
\begin{align*}
\nabla_X(\varphi^{\sharp}\circ\partial_x)=X(f)\cdot\varphi^{\sharp}\circ\nabla_{\partial_x}\partial_x=0
\end{align*}
We calculate the tension field with respect to an ON-frame as follows.
\begin{align*}
\tau(\varphi)&=\nabla_{e_j}d\varphi[Je_j]-d\varphi(\nabla_{e_j}Je_j)\\
&=\nabla_{e_j}\left(Je_j(f)(\varphi^{\sharp}\circ\partial_x)\right)-d\varphi\left(\scal[g]{\nabla_{e_j}Je_j}{e_k}Je_k\right)\\
&=\left(e_j\circ Je_j(f)-\scal[g]{\nabla_{e_j}Je_j}{e_k}Je_k(f)\right)(\varphi^{\sharp}\circ\partial_x)
\end{align*}
The statement is now immediate by (\ref{eqnFrameExpansion}) and (\ref{eqnLaplaceLocal}).
\end{proof}

\begin{Rem}
In terms of $f$, the superharmonic action functional considered in Sec. 4 of \cite{Gro13} reads
\begin{align*}
\mA(\varphi)=\frac{1}{2}\int_M\dsvol_g\,e_j(f)Je_j(f)
\end{align*}
\end{Rem}

In the following theorem, we will use the term ''vector field'' in the more general sense of vector field with flesh,
cf. \cite{Gro11a} for details. This is such that our previous results, in particular Prp. \ref{prpConservedQuantity}, continue to hold.

\begin{Thm}[Noether]
\label{thmSuperharmonicNoether}
Let $f$ be a superharmonic map $\laplace f=0$ and $\xi\in\mS M$ be a Killing vector field. Then the vector field
\begin{align*}
Y_{\xi}:=\frac{1}{2}e_j(f)Je_j(f)\xi-\xi(f)e_j(f)Je_j
\end{align*}
is divergence-free $\div Y_{\xi}=0$ and thus a conserved quantity.
\end{Thm}

\begin{proof}
By Prp. \ref{prpSuperharmonic}, the map $\varphi$ corresponding to $f$ is superharmonic.
Thm. 4.15 of \cite{Gro13} thus yields a divergence-free vector field $Y_{\xi}$, which we shall now express in terms of $f$.
Denoting the standard metric on $\bR$ by $g_{\bR}$, we yield
\begin{align*}
Y_{\xi}&=\left(\frac{1}{2}\scal[\varphi^*g_{\bR}]{e_j}{Je_j}\scal[g]{\xi}{e_i}-\scal[\varphi^*g_{\bR}]{\xi}{e_i}\right)Je_i\\
&=\left(\frac{1}{2}e_j(f)Je_j(f)\scal[g]{\xi}{e_i}-\xi(f)e_i(f)\right)Je_i
\end{align*}
which proves the statement.
\end{proof}

\section{Mean Value Theorems on the Supersphere}
\label{secMVT}

Let $B^m_L$ denote the open ball of radius $L>0$ in $\bR^m$ and $S^{m-1}_L=\partial B^m_L$ its boundary.
Their volumes are well-known to be
\begin{align}
\label{eqnVolume}
\mathrm{vol}(S^{m-1}_L)=2\frac{\pi^{m/2}}{\Gamma(\frac{m}{2})}L^{m-1}\;,\qquad
\mathrm{vol}(B^m_L)=\frac{\pi^{m/2}}{\Gamma(\frac{m}{2}+1)}L^m
\end{align}
Let $f$ be a harmonic function defined on a domain $\Omega\subseteq\bR^m$ such that $\overline{B^m_L}\subseteq\Omega$.
The classical mean value theorem (see Sec. 2.2.2 of \cite{Eva10}) can then be stated as follows.
\begin{align}
\label{eqnClassicalMVT}
\int_{S_L^{m-1}}\dvol_{\tilde{g}}\,\iota_L^*(f)=\mathrm{vol}(S^{m-1}_L)\cdot f(0)\;,\qquad
\int_{B_L^m}\dvol_g\,f=\mathrm{vol}(B^m_L)\cdot f(0)
\end{align}
where $g$ denotes the standard metric on $\bR^m$ and $\tilde{g}:=\iota_L^*g$ its pullback to $S^{m-1}_L$
under the inclusion $\iota_L:S^{m-1}_L\rightarrow\bR^m$.

The $L$-supersphere in $\bR^{m|2n}$ was introduced by formally setting a certain extension of the radius
to $L$. Integration over the supersphere was introduced by an extension of the Pizzetti formula,
for super-polynomials in \cite{DBS07} and later for general superfunctions in \cite{CDS09},
and finally expressed in terms of an embedding of the supersphere in \cite{Cou12}.
It was found that (\ref{eqnVolume}) continues to hold in that context with $m$ replaced by the so called
superdimension $M:=m-2n$. In \cite{CDS10}, a super analogon of the mean value theorem (\ref{eqnClassicalMVT}) for
harmonic superfunctions was established.

The first purpose of this final section is to formulate the supersphere integration mentioned
in terms of a retraction $\gamma$ along the lines of the general theory of Sec. \ref{secIntegration}
and to reproduce the volume formulas. In Sec. \ref{subsecMVT}, we will finally give a self-contained proof
of two mean value theorems, which is based on the divergence Thm. \ref{thmSuperDivergence} in a very natural way,
thus avoiding the subtleties left open in the proof of \cite{CDS10}.

\subsection{Flat Superspace and the Supersphere}

On $\bR^{m|2n}$, consider standard global coordinates $(x,\theta)$ and, with respect to the induced
coordinate vector fields, the supermetric from \cite{CDS10} as follows.
\begin{align}
\label{eqnSuperspaceMetric}
\renewcommand{\arraystretch}{1.5}
g&:=\left(\arr{@{\;}c@{\;}|@{\;}c@{\;}}{G_{0,m}&0\\\hline 0&-\frac{1}{2}J_{2n}}\right)
\end{align}
with $G_{0,m}$ and $J_{2n}$ as in (\ref{eqnStandardMetric}). The volume form then reads
\begin{align}
\label{eqnSuperspaceVolumeform}
\dsvol_g=[d(x,\theta)]\,2^n
\end{align}
Moreover,
\begin{align}
\label{eqnONFrame}
e_k:=\left\{\arr{ll}{\dd{x^k}&k\leq m\\\dd{\theta^k}&k=m+2l-1\\-2\dd{\theta^k}&k=m+2l}\right.
\end{align}
is seen to be an $OSp_{(0,m)|2n}$-frame. By means of (\ref{eqnLaplaceLocal}), the Laplacian thus becomes
\begin{align}
\label{lemLaplaceSupersphere}
\laplace f=-e_k\circ Je_k(f)=-\left(\sum_k(\partial_{x^k})^2-4\sum_{j\,\mathrm{odd}}\partial_{\theta^j}\partial_{\theta^{j+1}}\right)f
\end{align}
The superradius as considered in \cite{CDS10} can be understood to be obtained by formally identifying the tuple of standard coordinates $\vx:=(x^1,\ldots,\theta^1,\ldots)$
as tuple of coefficients of the vector fields $(\dd{x^1},\ldots,\dd{\theta^1},\ldots)$, i.e. as
$\vx=\sum_ix^i\cdot\dd{x^i}+\sum_j\theta^j\cdot\dd{\theta^j}$, such that
\begin{align}
\label{eqnSuperradius}
R^2:=\scal[g]{\vx}{\vx}=\sum_i(x^i)^2-\sum_{j\,\mathrm{odd}}\theta^j\theta^{j+1}=:r^2+\vartheta^2
\end{align}
For the $n$-th power of $\vartheta^2$, we find
\begin{align}
\label{eqnThetaPower}
\vartheta^{2n}=\left(-\sum_{j\,\mathrm{odd}}\theta^j\theta^{j+1}\right)^n
=(-1)^nn!\theta^1\theta^2\cdot\theta^{2n}
\end{align}
Let $L>0$. The sphere $S^{m-1}_L$ with radius $L$ can be defined by the
boundary function $\rho:\bR^m\rightarrow\bR$, $\rho:=r$ becoming $L$
and is the boundary $S^{m-1}_L=\partial B^m_L$ of the $L$-ball.
Similarly, we define the superfunction $\tau\in\mO_{\bR^{m|2n}}$ by $\tau:=R$.
Consider the retraction $\gamma:\bR^{m|2n}\rightarrow\bR^m$ as follows.
\begin{align}
\label{eqnGamma}
\gamma^*(x^j):=x^j\sqrt{1+\frac{\vartheta^2}{r^2}}\qquad\textrm{such that}\qquad\tau=\gamma^*(\rho)
\end{align}
The $L$-superball $B^{m|2n}_L:=\bR^{m|2n}|_{B^m_L}$ is defined to be the superspace $\bR^{m|2n}$ restricted to the open ball $B^m_L$.
By (\ref{eqnBoundaryDiagram}), there is a unique boundary supermanifold $\partial_{\gamma}B^{m|2n}_L$
with immersion and a unique compatible boundary retraction, respectively, as follows.
\begin{align*}
\iota_L:=\iota_{\gamma}:\partial_{\gamma}B^{m|2n}_L\rightarrow\bR^{m|2n}\quad\textrm{such that}\quad\iota_{\gamma}^*(\tau)=L\;,\qquad
\gamma_L:=\partial\gamma:\partial_{\gamma}B^{m|2n}_L\rightarrow S_L^{m-1}
\end{align*}

\begin{Def}
The supermanifold $S^{m-1|2n}_L:=\partial_{\gamma}B^{m|2n}_L$ is called the $L$-supersphere.
When the dimension is understood, we shall also abbreviate $S_L:=S^{m-1|2n}_L$.
\end{Def}

The maps involved are best described by means of spherical coordinates
\begin{align}
\label{eqnSphericalCoordinates}
(r,\phi^0,\phi^1,\ldots,\phi^{m-2}):(0,\infty)\times(0,2\pi)\times(0,\pi)^{m-2}\rightarrow\bR^m
\end{align}
of $\bR^m$ (with a nullset removed). Similarly, the following are two different tuples of coordinates
of $\bR^{m|2n}$ (with the same nullset removed).
\begin{align}
\label{eqnSuperSphericalCoordinates}
(r,\phi^0,\ldots,\phi^{m-2},\theta^1,\ldots,\theta^{2n})\qquad\mathrm{or}\qquad
(R,\phi^0,\ldots,\phi^{m-2},\theta^1,\ldots,\theta^{2n})
\end{align}
The second tuple of coordinates has the form $x=(\tau,\tilde{x})$
as described in the paragraph preceding (\ref{eqnBoundaryVectorFields}) with $\tau=R$.
The coordinate transformation $\varphi=(\id,\varphi^{\sharp})$ between these coordinates takes the following form.
\begin{align}
\label{eqnPhi}
\varphi^*(r)=\sqrt{R^2-\vartheta^2}\,,\quad\varphi^*(\phi^k)=\phi^k\,,\quad\varphi^*(\theta^j)=\theta^j
\end{align}
such that $\varphi^*(r^2+\vartheta^2)=R^2$, while $\gamma$ becomes
\begin{align}
\label{eqnSphericalGamma}
\gamma^*(r)=R\;,\qquad\gamma^*(\phi^k)=\phi^k\;,\qquad\gamma^*(\theta^k)=\theta^k\qquad\textrm{such that}\quad\gamma_L^*(\phi^k)=\phi^k
\end{align}
Let $(y,\theta)$ denote coordinates of $\bR^{m|2n}$ which depend on the coordinates $(R,\phi,\theta)$
in the same way as the standard coordinates $(x,\theta)$ depend on the spherical coordinates $(r,\phi,\theta)$
of (\ref{eqnSuperSphericalCoordinates}). Then the map $\varphi$ becomes
\begin{align*}
\varphi^*(x^j)=y^j\sqrt{1-\frac{\vartheta^2}{R^2}}\;,\qquad\varphi^*(\theta^j)=\theta^j
\end{align*}
which is as considered in Sec. 5.2 of \cite{Cou12} (where $y^j$ and $R$ are denoted $x^j$ and $r$, respectively).
It follows that $\varphi^*\circ\gamma^*(x^j)=y^j$ and, in this sense, the diagram
\begin{align}
\label{eqnRetractionDiagram}
\begin{xy}
\xymatrix{\bR^{m|2n}\ar[dr]_{\gamma_{\mathrm{std}}}\ar[rr]^{\varphi}&\ar@{}[d]|{\circlearrowleft}&\bR^{m|2n}\ar[dl]^{\gamma}\\
&\bR^m}
\end{xy}
\end{align}
which is defined up to nullsets, commutes.

\begin{Lem}
\label{lemOuterNormal}
The vector field $\partial_{\tau}$ and the outer normal vector field $\nu$ on $S_L^{m-1|2n}$ are, respectively,
\begin{align*}
\partial_{\tau}=\frac{R}{r}\partial_r\;,\qquad
\nu=\frac{r}{R}\left(\partial_r+\frac{\theta^i}{r}\partial_{\theta^i}\right)
\end{align*}
with respect to the first coordinates in (\ref{eqnSuperSphericalCoordinates}) and $R=\sqrt{r^2+\vartheta^2}$.
They satisfy
\begin{align*}
\varepsilon(\nu)=1\;,\qquad\scal[g]{\partial_{\tau}}{\nu}=1\;,\qquad\nu(R)=1
\end{align*}
\end{Lem}

Remark that $R\nu=\mathbb{E}$ with $\mathbb{E}$ in \cite{CDS10}.

\begin{proof}
With respect to the second coordinates in (\ref{eqnSuperSphericalCoordinates}), the tangent sheaf of $S^{m-1|2n}_L$
is spanned by the vector fields $\partial_{\phi^k}$ and $\partial_{\theta^k}$
while $\partial_{\tau}=\partial_R$. Transferred to extended
spherical coordinates via (\ref{eqnPhi}) as in (\ref{eqnBoundaryVectorFields}) and using the calculation rule
(\ref{eqnCalculationRule}), these vector fields become
\begin{align*}
\widehat{\partial_R}:=(\varphi^{\sharp})^{-1}\circ\partial_R\circ\varphi^{\sharp}
&=(\varphi^{\sharp})^{-1}\left(\partial_R\varphi^{\sharp}(\zeta^i)(\varphi^{\sharp}\circ\partial_{\zeta^i})\right)\\
&=(\varphi^{\sharp})^{-1}\left(\frac{R}{\sqrt{R^2-\vartheta^2}}\right)\partial_r=\frac{\sqrt{r^2+\vartheta^2}}{r}\cdot\partial_r
\end{align*}
and
\begin{align*}
\widehat{\partial_{\phi^k}}:=(\varphi^{\sharp})^{-1}\circ\partial_{\phi^k}\circ\varphi^{\sharp}=\partial_{\phi^k}\;,\qquad
\widehat{\partial_{\theta^k}}:=(\varphi^{\sharp})^{-1}\circ\partial_{\theta^k}\circ\varphi^{\sharp}
=\left\{\arr{ll}{\partial_{\theta^k}+\frac{1}{2}\frac{\theta^{k+1}}{r}\partial_r&k\;\mathrm{odd}\\
\partial_{\theta^k}-\frac{1}{2}\frac{\theta^{k-1}}{r}\partial_r&k\;\mathrm{even}}\right.
\end{align*}
A straightforward calculation using (\ref{eqnSuperspaceMetric}) yields that $\nu$ as stated satisfies
\begin{align*}
\scal[g]{\nu}{\widehat{\partial_{\phi^k}}}=0\;,\qquad
\scal[g]{\nu}{\widehat{\partial_{\theta^k}}}=0\;,\qquad
\scal[g]{\nu}{\nu}=1\;,\qquad
\scal[g]{\partial_{\tau}}{\nu}=1
\end{align*}
and is thus the outer normal as claimed. The last statement follows from a simple calculation.
\end{proof}

As an immediate corollary of Lem. \ref{lemInducedMetric}, we obtain the following.

\begin{Cor}
\label{corInducedMetric}
The volume form with respect to the pullback metric $\tilde{g}=\iota_L^*g$ is the restriction
$\dsvol_{\tilde{g}}=\dsvol_g|_{S^{m-1|2n}_L,\tau}$.
\end{Cor}

The change of integral forms that enters the transformation formula for spherical coordinates reads
\begin{align}
\label{eqnSphericalDeterminant}
\frac{[dx]}{[d(r,\phi)]}=r^{m-1}\cdot\Omega\;,\qquad
\Omega:=(\sin\phi^{m-2})^{m-2}\cdot(\sin\phi^{m-3})^{m-3}\cdot\ldots\cdot(\sin\phi^1)
\end{align}
With this notation, the integral over the sphere $S^{m-1}$ becomes
\begin{align*}
\int d\Omega:=\int_0^{2\pi}d\phi^0\int_0^{\pi}d\phi^1\ldots\int_0^{\pi}d\phi^{m-2}\,\Omega
\end{align*}
Similarly, the change of variables corresponding to (\ref{eqnPhi}) induces the transformation (\ref{eqnBerezinTransformation}) as follows.
\begin{align}
\label{eqnBerezinSuperspherical}
[d(r,\phi,\theta)]=[d(R,\phi,\theta)]\,\dd[\phi^*(r)]{R}
=[d(R,\phi,\theta)]\,\frac{R}{\sqrt{R^2-\vartheta^2}}
\end{align}
For the integration with respect to $\gamma$, we thus obtain the following formula.

\begin{Lem}
\label{lemSuperballIntegration}
Let $f\in\mO_{\bR^{m|2n}}(B^m_L)$ be a function defined on the $L$-ball, and
$\gamma$ and $\varphi$ be the retraction and coordinate transformation
of (\ref{eqnGamma}) and (\ref{eqnPhi}), respectively. Then
\begin{align*}
\int_{(B^{m|2n}_L,\gamma)}\dsvol_g\,f=2^n\int d^{2n}\theta\int_0^L dR\int d\Omega\,R\left(R^2-\vartheta^2\right)^{\frac{m-2}{2}}\varphi^*(f)
\end{align*}
\end{Lem}

\begin{proof}
We calculate, using (\ref{eqnSuperspaceVolumeform}), (\ref{eqnSphericalDeterminant}) and (\ref{eqnBerezinSuperspherical}),
\begin{align*}
\int_{(B^{m|2n}_L,\gamma)}\dsvol_g\,f&=2^n\int_{(B^{m|2n}_L,\gamma)}[d(r,\phi,\theta)]\,r^{m-1}\Omega\cdot f\\
&=2^n\int_{(B^{m|2n}_L,\gamma)}[d(R,\phi,\theta)]\,\frac{R}{\sqrt{R^2-\vartheta^2}}\varphi^*(r^{m-1}\Omega\cdot f)\\
&=2^n\int_{(B^{m|2n}_L,\gamma)}[d(R,\phi,\theta)]\,R(R^2-\vartheta^2)^{\frac{m-2}{2}}\Omega\,\varphi^*(f)
\end{align*}
Now by (\ref{eqnSphericalGamma}), the retraction $\gamma$ simply renames the coordinate $r$ to $R$
(which corresponds to commutation of the diagram in (\ref{eqnRetractionDiagram}),
from which the statement follows.
\end{proof}

\begin{Lem}
\label{lemEpsilonSphereIntegral}
Let $f\in\mO_{\bR^{m|2n}}(\Omega)$ be a function defined on an open neighbourhood $\Omega_0\subseteq\bR^m$ of $S^{m-1}_L$. Then
\begin{align*}
\int_{(S_L,\gamma_L)}\dsvol_{\tilde{g}}\,\iota_L^*(f)
=2^n\int d^{2n}\theta\int d\Omega\,L\left(L^2-\vartheta^2\right)^{\frac{m-2}{2}}\iota_L^*(f)
\end{align*}
\end{Lem}

\begin{proof}
By Cor. \ref{corInducedMetric} and (\ref{eqnSuperspaceVolumeform}), we find
\begin{align*}
\int_{(S_L,\gamma_L)}\dsvol_{\tilde{g}}\,\iota_L^*(f)
&=2^n\int_{(S_L,\gamma_L)}[d(x,\theta)]|_{S_L,\tau}\,\iota_L^*(f)\\
&=2^n\int_{(S_L,\gamma_L)}\left([d(R,\phi,\theta)]\,
\frac{[d(r,\phi,\theta)]}{[d(R,\phi,\theta)]}
\frac{[d(x,\theta)]}{[d(r,\phi,\theta)]}\right)|_{S_L,\tau}\,\iota_L^*(f)\\
&=2^n\int_{(S_L,\gamma_L)}[d\iota_L^*(\phi,\theta)]\,\iota_L^*\left(\frac{R}{\sqrt{R^2-\vartheta^2}}r^{m-1}\Omega\right)\iota_L^*(f)\\
&=2^n\int_{(S_L,\gamma_L)}[d(\phi,\theta)]\,L(L^2-\vartheta^2)^{\frac{m-2}{2}}\Omega\,\iota_L^*(f)
\end{align*}
As in the proof of the previous lemma, the statement now follows from (\ref{eqnSphericalGamma}).
\end{proof}

Lem. \ref{lemSuperballIntegration} and Lem. \ref{lemEpsilonSphereIntegral} immediately yield the following Cavalieri result.

\begin{Cor}
\label{corEpsilonSphereIntegral}
Integration of $f\in\mO_{\bR^{m|2n}}(B^m_L)$ over the $L$-superball is related with supersphere integration as follows.
\begin{align*}
\int_{(B_L,\gamma)}\dsvol_g\,f=\int_0^LdR\int_{(S_R,\gamma_R)}\dsvol_{\tilde{g}}\,\iota_R^*(f)
\end{align*}
\end{Cor}

\begin{Lem}
\label{lemVolumeSupersphere}
The area of the $L$-supersphere is
\begin{align*}
\mathrm{vol}(S^{m-1|2n}_L):=\int_{(S_L,\gamma_L)}\dsvol_{\tilde{g}}=(-1)^{s(m,2n)}\frac{2^{n+1}\pi^{m/2}}{\Gamma(\frac{M}{2})}L^{M-1}
\end{align*}
with $M:=m-2n$. In case $M=-2l$ with $l\in\bN_0$, where the Gamma function becomes infinite, this has to be
read as $\mathrm{vol}(S^{m-1|2n})=0$.
\end{Lem}

Comparing this result to the supersphere area in \cite{DBS07}, note the different convention $s(m,2n)=1$
there opposed to (\ref{eqnBerezinSign}) and the additional factor of $\pi^{-n}$ in (\ref{eqnBerezinIntegral}).
Moreover, considering $[d(x,\theta)]$ instead of $\dsvol_g$ removes a factor of $2^n$, thus yielding the volume of the ordinary sphere
(\ref{eqnVolume}) with $m$ replaced by $M$.

\begin{proof}
By Lem. \ref{lemEpsilonSphereIntegral}, we find
\begin{align*}
\mathrm{vol}(S^{m-1|2n}_L)&=2^nL^{m-1}\int d^{2n}\theta\int d\Omega\,\left(1-\frac{\vartheta^2}{L^2}\right)^{\frac{m-2}{2}}
\end{align*}
The integrand can be expressed by the following Taylor formula with $g(y):=y^{\frac{m-2}{2}}$.
\begin{align}
\label{eqnPochhammerTaylor}
\left(1-\frac{\vartheta^2}{L^2}\right)^{\frac{m-2}{2}}
=\sum_{k=0}^n\frac{1}{k!}\cdot\partial_y^{(k)}g(1)\cdot\left(-\frac{\vartheta^2}{L^2}\right)^k
\end{align}
where
\begin{align}
\label{eqnPochhammer}
\partial_y^{(k)}g(1)=\left(\frac{m}{2}-1\right)\left(\frac{m}{2}-2\right)\cdot\ldots\cdot\left(\frac{m}{2}-k\right)
\end{align}
In the Berezin integral, only the summand with $k=n$ remains.

Case 1: $M=m-2n=-2l$ with $l\in\bN_0$. Then $m$ is even and $m\leq 2n$ or, equivalently,
$\frac{m-2}{2}\leq n-1$ such that $\partial_y^{(n)}g=0$. In this case, we thus find
$\mathrm{vol}(S^{m-1|2n})=0$.

Case 2: $M=2+2l$ with $l\in\bN_0$. Then $m$ is even and $m-2n\geq 2$ or, equivalently, $\frac{m-2}{2}\geq n$.
With (\ref{eqnThetaPower}), we find that
\begin{align*}
\mathrm{vol}(S^{m-1|2n}_L)&=2^n\frac{L^{m-1}}{L^{2n}}\left(\frac{m}{2}-1\right)\cdot\ldots\cdot\left(\frac{m}{2}-n\right)\int d^{2n}\theta\int d\Omega\,\theta^1\theta^2\cdot\theta^{2n}\\
&=(-1)^{s(m,2n)}2^nL^{M-1}\left(\frac{m}{2}-1\right)\cdot\ldots\cdot\left(\frac{m}{2}-n\right)\mathrm{vol}(S^{m-1})
\end{align*}
Now using (\ref{eqnVolume}) gives the expression claimed. The calculation is similar in the remaining two cases,
which we state for the sake of completeness.

Case 3: $M=1+2l$ with $l\in\bN_0$. Then $m$ is odd and $m-2n\geq 1$, and (\ref{eqnPochhammer})
with $k=n$ consists of half-integer factors and is positive.

Case 4: $M=-1-2l$ with $l\in\bN_0$. Then $m$ is negative and $m-2n\leq -1$. This time,
(\ref{eqnPochhammer}) with $k=n$ consists of positive and negative factors.
\end{proof}

\begin{Cor}
\label{corVolumeSupersphere}
The volume of the $L$-superball is
\begin{align*}
\mathrm{vol}(B_L^{m|2n}):=\int_{(B_L,\gamma)}\dsvol_g=(-1)^{s(m,2n)}\frac{2^n\pi^{m/2}}{\Gamma(\frac{M}{2}+1)}L^M
\end{align*}
In case $M=-2l$ with $l\in\bN_0$, both sides vanish.
\end{Cor}

\begin{proof}
This is calculated with Cor. \ref{corEpsilonSphereIntegral} and Lem. \ref{lemVolumeSupersphere} as follows.
\begin{align*}
\int_{(B_L,\gamma)}\dsvol_g
=\int_0^LdR\,\mathrm{vol}(S^{m-1|2n}_R)
=(-1)^{s(m,2n)}\frac{2^{n+1}\pi^{m/2}}{\Gamma(\frac{M}{2})}\int_0^LdR\,R^{M-1}
\end{align*}
and the statement follows from $\frac{M}{2}\Gamma\left(\frac{M}{2}\right)=\Gamma\left(\frac{M}{2}+1\right)$.
\end{proof}

\subsection{Mean Value Theorems}
\label{subsecMVT}

Before coming to the first mean value Thm. \ref{thmMVT} below,
we state two auxiliary results. In the following, $\gamma$ continues to denote the retraction
(\ref{eqnGamma}), $M:=m-2n$ the superdimension and $R$ the superradius (\ref{eqnSuperradius}).
As usual, $\log R$ is defined by means of the Taylor expansion with respect to the odd coordinates.

\begin{Lem}
\label{lemFundamentalSolution}
Consider a bounded domain $\Omega\subseteq\bR^{m|2n}$ with $\Omega_0$ not containing $0\in\bR^m$.
Restricted to $\Omega$, the following holds.
\begin{align*}
\laplace R^{2-M}=0\quad\mathrm{for}\quad M\neq 2\quad,\qquad\laplace\log R=0\quad\mathrm{for}\quad M=2
\end{align*}
\end{Lem}

\begin{proof}
This follows from a direct calculation using (\ref{lemLaplaceSupersphere}) and (\ref{eqnCalculationRule}).
\end{proof}

Minor modifications of the superfunctions in the preceding lemma may be interpreted as fundamental solutions
of the Laplacian $\laplace$ as in Sec. 4 of \cite{CDS10}. In fact, the proof of Thm. 6 in \cite{CDS10} seems to make implicit use of a divergence theorem
for super distributions applied to those fundamental solutions. Our proof (of Thm. \ref{thmMVT} below) avoids to
make these subtleties rigorous by performing a direct limit argument instead.

\begin{Lem}
\label{lemDivergenceRadius}
Fix $L>0$. Let $f,k\in\mO_{\bR^{m|2n}}(\Omega)$ be functions defined on a domain $\Omega_0\subseteq\bR^m$
with $\overline{B^m_L}\subseteq\Omega_0$ such that $\laplace f\equiv 0$ and $\iota_L^*(k)=C\in\bR$. Then
\begin{align*}
\int_{(B^{m|2n}_L,\gamma)}\dsvol_g\,\div(k\nabla f)=0
\end{align*}
Let $0<\varepsilon<L$. If, in addition, $\iota_{\varepsilon}^*(k)=C'\in\bR$, then the statement continues to hold
for $B_L$ replaced by $B_L\setminus B_{\varepsilon}$. In this case, $f$ and $k$ need not be defined near $0$.
\end{Lem}

\begin{proof}
The integral vanishes due to the following calculation, using Thm. \ref{thmSuperDivergence} twice.
\begin{align*}
\int_{(B^{m|2n}_L,\gamma)}\dsvol_g\,\div(k\nabla f)
&=S\int_{(S_L,\gamma_L)}\dsvol_{\tilde{g}}\,\iota_L^*(k)\iota_L^*\left(\scal[g]{\nabla f}{\nu}\right)\\
&=C\int_{(B_L,\gamma)}\dsvol_g\,\div(\nabla f)\\
&=0
\end{align*}
The second part of the statment follows by an analogous calculation over both boundary parts $S_L$ and $S_{\varepsilon}$.
\end{proof}

\begin{Thm}[Mean Value Theorem]
\label{thmMVT}
Let $f\in\mO_{\bR^{m|2n}}(\Omega)$ be a function defined on a domain $\Omega_0\subseteq\bR^m$
with $\overline{B^m_L}\subseteq\Omega_0$ which is harmonic $\laplace f=0$. Then
\begin{align*}
\int_{(S^{m-1|2n}_L,\gamma_L)}\dsvol_{\tilde{g}}\,\iota_L^*(f)=\mathrm{vol}(S^{m-1|2n}_L)\cdot f(0)
\end{align*}
with $\mathrm{vol}(S^{m-1|2n}_L)$ as in Lem. \ref{lemVolumeSupersphere}.
\end{Thm}

\begin{proof}
We consider the case $M\neq 2$ first. Using Lem. \ref{lemOuterNormal}, we yield
\begin{align*}
\iota_L^*\scal[g]{\nabla R^{2-M}}{\nu}=\iota_L^*\left(\nu(R^{2-M})\right)=(2-M)L^{1-M}
\end{align*}
and, therefore,
\begin{align*}
&(2-M)L^{1-M}\int_{(S_L,\gamma_L)}\dsvol_{\tilde{g}}\,\iota_L^*(f)\\
&\qquad=\int_{(S_L,\gamma_L)}\dsvol_{\tilde{g}}\,\iota_L^*\left(\scal[g]{f\nabla R^{2-M}}{\nu}\right)\\
&\qquad=S\int_{(B_L\setminus B_{\varepsilon},\gamma)}\dsvol_g\,\div(f\nabla R^{2-M})
+\int_{(S_{\varepsilon},\gamma_{\varepsilon})}\dsvol_{\tilde{g}}\,\iota_{\varepsilon}^*\left(\scal[g]{f\nabla R^{2-M}}{\nu}\right)\\
&\qquad=I_{\varepsilon}+II_{\varepsilon}
\end{align*}
by using the divergence Thm. \ref{thmSuperDivergence} for $0<\varepsilon<1$.
The first term vanishes due to the following calculation.
By (\ref{eqnLaplaceDifference}) together with Lem. \ref{lemFundamentalSolution} and the assumption $\laplace f=0$, we find
\begin{align*}
I_{\varepsilon}=S\int_{(B_L\setminus B_{\varepsilon},\gamma)}\dsvol_g\,\div(R^{2-M}\nabla f)=0
\end{align*}
which vanishes according to Lem. \ref{lemDivergenceRadius}.
For the second part, we have by Lem. \ref{lemEpsilonSphereIntegral}
\begin{align*}
II_{\varepsilon}&=(2-M)\varepsilon^{1-M}\int_{(S_{\varepsilon},\gamma_{\varepsilon})}\dsvol_{\tilde{g}}\,\iota_{\varepsilon}^*(f)\\
&=(2-M)\varepsilon^{1-M}2^n\int d^{2n}\theta\int d\Omega\,\varepsilon^{m-1}\left(1-\frac{\vartheta^2}{\varepsilon^2}\right)^{\frac{m-2}{2}}\iota_{\varepsilon}^*(f)\\
&=(2-M)\varepsilon^{2n}2^n\int d^{2n}\theta\int d\Omega\,\left(1-\frac{\vartheta^2}{\varepsilon^2}\right)^{\frac{m-2}{2}}\iota_{\varepsilon}^*(f)
\end{align*}
Consider the Taylor expansion (\ref{eqnPochhammerTaylor}) with $L=\varepsilon$.
The only summand that does not vanish in the limit $\varepsilon\rightarrow 0$
is that with $k=n$, where $\varepsilon^{2n}$ in the denominator cancels with the same term in front of the integral.
It follows, moreover, that all higher terms of the Graßmann expansion of $\iota_{\varepsilon}^*(f)$
cancel with the term $\vartheta^{2n}$ occuring in the previously mentioned $k=n$ summand.
In the limit, the remaining part of $\iota_{\varepsilon}^*(f)$ thus goes to $f(0)$,
and we may apply the calculations in the proof of Lem. \ref{lemVolumeSupersphere} to obtain
\begin{align*}
\lim_{\varepsilon\rightarrow 0}II_{\varepsilon}
=(2-M)(-1)^{s(m,2n)}\frac{2^{n+1}\pi^{m/2}}{\Gamma(\frac{M}{2})}\cdot f(0)=\frac{(2-M)}{L^{M-1}}\mathrm{vol}(S^{m-1|2n}_L)\cdot f(0)
\end{align*}
from which the statement for $M\neq 2$ follows.

The case $M=2$ is shown analogous with $R^{2-M}$ replaced by $\log R$.
\end{proof}

\begin{Thm}[Mean Value Theorem]
\label{thmMVT2}
Let $f\in\mO_{\bR^{m|2n}}(\Omega)$ be a function defined on a domain $\Omega_0\subseteq\bR^m$
with $\overline{B^m_L}\subseteq\Omega_0$ which is harmonic $\laplace f=0$. Then
\begin{align*}
\int_{(B_L,\gamma)}\dsvol_g\,f&=\mathrm{vol}(B_L^{m|2n})\cdot f(0)\\
\int_{(B_L,\gamma_{\mathrm{std}})}\dsvol_g\,f&=(-1)^{s(m,2n)}2^n\mathrm{vol}(B^m_L)\cdot f^{(1,\ldots,1)}(0)
\end{align*}
with $\mathrm{vol}(B_L^{m|2n})$ as in Cor. \ref{corVolumeSupersphere}.
\end{Thm}

\begin{proof}
By Cor. \ref{corEpsilonSphereIntegral} and Thm. \ref{thmMVT}, we yield
\begin{align*}
\int_{(B_L,\gamma)}\dsvol_g\,f
=\int_0^LdR\int_{(S_R,\gamma_R)}\dsvol_{\tilde{g}}\,\iota_R^*(f)
=\int_0^LdR\,\mathrm{vol}(S^{m-1|2n}_R)\cdot f(0)
\end{align*}
The first statement then follows as in Cor. \ref{corVolumeSupersphere}.
For the second, note that $\laplace f=0$ implies that $f^{(1,\ldots,1)}$ is harmonic in the classical sense.
We may thus apply the classical mean value theorem (\ref{eqnClassicalMVT}) after the following calculation.
\begin{align*}
\int_{(B_L,\gamma_{\mathrm{std}})}\dsvol_g\,f
=2^n\int_{B_L}d^mx\int d^{2n}\theta\,f
=(-1)^{s(m,2n)}2^n\int_{B_L}d^mx\,f^{(1,\ldots,1)}
\end{align*}
\end{proof}

As an immediate corollary, we find the following simple expression for the boundary term of
Thm. \ref{thmChangeOfRetraction}, which usually has the complicated form stated there.

\begin{Cor}
In the situation of Thm. \ref{thmMVT2}, the boundary term equals
\begin{align*}
b_{B,\gamma_{\mathrm{std}},\gamma}(\dsvol_g\,f)
=(-1)^{s(m,2n)}2^n\mathrm{vol}(B^m_L)\cdot f^{(1,\ldots,1)}(0)-\mathrm{vol}(B_L^{m|2n})\cdot f(0)
\end{align*}
\end{Cor}

\section*{Acknowledgements}

I would like to thank Florian Hanisch for interesting discussions.

\bibliographystyle{alpha}

\end{document}